\documentclass{amsart}

\usepackage[width=17cm, left=2cm, height=23cm]{geometry}
\usepackage{multirow}
\usepackage{rotating}

\theoremstyle{plain}
\newtheorem{theorem}{Theorem}[section]
\newtheorem{lemma}[theorem]{Lemma}

\theoremstyle{remark}
\newtheorem{remark}[theorem]{Remark}

\theoremstyle{definition}
\newtheorem{definition}[theorem]{Definition}

\newcommand{\R}{\mathbb R}

\newcommand{\Ff}{\mathcal{F}}

\begin{document}

\title[Detect a cheater's action]{A randomized most powerful test to detect a cheater's action. Applicaton to identification of listeriosis in Lombardy}


\author[G. Aletti]{Giacomo Aletti}
\address{ADAMSS Center and Dept. of Mathematics, Universit\`a degli Studi di Milano, Via Saldini 50, 20133, Milano, Italy}
\email{giacomo.aletti@unimi.it}

\author{Irene Matuonto}
\address{ADAMSS Center and Dept. of Mathematics, Universit\`a degli Studi di Milano, Via Saldini 50, 20133, Milano, Italy}
\email{iremat@libero.it}
\thanks{corresponding author}

\author{Mirella Pontello}
\address{Dept. of Health Sciences, Universit\`a degli Studi di Milano, via di Rudin\`\i\ 8, 20143 Milano , Italy}
\email{mirella.pontello@unimi.it}

\subjclass[2010]{Primary: 62P10; Secondary: 62G05}
 
\keywords{monotone models; identification of epidemics; rare events testing; extremeness; listeriosis}

\date{October 12, 2014}


\begin{abstract}
This article presents a new randomized non-parametric test based on a sample of independent but not identically distributed variables; this test detects if a cheater replaces one of the distributions of the sample with a convex-dominating one.
The presented test is the uniformely most powerful, in the sense that it is the most powerful for any change of the cheater.
We show that this test may be applied when we have variables with distribution satisfying the monotone likelihood ratio property and we need to check whether a parameter of a variable has been changed.\\
The application we present concerns the detection of epidemics of listeriosis in Lombardy from 2005 to 2011.
\end{abstract}

\maketitle

\section{Introduction}

Let us consider a sample of independent random variables with (0,1)-uniform distribution:
$Y_1, Y_2, \ldots, Y_n$. We need to discover if a cheater has replaced one of
the variables with another r.v. $Y$ defined on $(0,1)$ with a convex cdf $F_Y$, stochastically dominating
the uniform ditribution, i.e.\ $F(x)\geq x, \forall x\in (0,1) .$\\
We test
the null hypothesis $H_0$ of equal and uniform distribution of all the variables against
the alternative hypothesis $H_1$ of a cheater's replacement. We choose to observe the statistics
$\hat Y=\max(Y_1,\ldots, Y_n)$, that we call extreme event.
In Theorem \ref{teo: MP_test} we present the non-parametric uniformly most powerful (UMP) test,
in the sense that it is the most powerful (MP) test for any cheater's choice; in particular, if
$\alpha$ is the significance level of the test, we reject $H_0$ if $\hat Y > \sqrt[n]{1-\alpha}$.\\
This game can be extended to a sample of independent r.v.'s $X_1, X_2, \ldots, X_n$
with cfd $F^0_{X_1}, F^0_{X_2}, \ldots, F^0_{X_n}$ by noting that $Y_i = F^0_{X_i}(X_i)\sim U(0,1)$; if
the original distributions are discrete, a randomization can be applied (see Lemma \ref{lem:unif}).
In this general framework, we suppose that the cheater changes $F^0_{X_j}$ with another distribution $F^1_{X_j}$: to
apply the previous results we should be sure that the randomization of $F^0_{X_j}(X_j)$ is convex when $X_j\sim F^1_{X_j}$.
We show that this is the case whenever $F^0_{X_j}$ and $F^1_{X_j}$ have the monotone likelihood ratio property
(see Theorem~\ref{teo: MLR}).\\
As an example, the result can be applied when we ask ourselves whether the occurrences of a (even rare)
disease show evidence of an epidemics. In our case, we apply it to a listeriosis database.

The first important work on epidemics based on spatial analysis is
Dr.~John Snow's study of London's cholera epidemics \cite{snow}. After
him, many other researchers have used spatial analysis to study
subjects concerning Public Health; a wide collection of these
topics, with particular interest to the statistical point of view,
is given by Waller and Gotway \cite{pubhealth}.\\
In our data, we have a sample of random variables $N_{i,j}$ counting the number of occurrences
of a disease in a certain region $i$ at time $j$ that we model with a
Poisson distribution, so $F^0_{i,j}\sim \mathcal{P}(\lambda \cdot p_{i,j})$, i.e.
the number $N_{i,j}$ is distributed as a Poisson depending
on a parameter $\lambda$ and on the population $p_{i,j}$ of region $i$ at time $j$.
When an epidemics occurs, the parameter $\lambda$ of the corresponding $N_{i,j}$
increases: $\lambda_1 > \lambda$. Since the family of Poisson model
has the monotone likelihood ratio property, we may apply our test to
$(F^0_{i,j}(N_{i,j}))_{i,j}$.\\
The small number of data is the critical point of
the analysis, because usual limit theorems cannot be used; for this
reason it is important to find the UMP test.\\

In Section \ref{sect: MP_test} we present the main result of the
article: the UMP test for extreme monotone randomized models. First of
all we recall the definition and some properties of the Skorohod
representation of a random variable, then in Subsection
\ref{subsect: MP1} we introduce the randomization of the quantile
function for a non-continuous random variable, then in Subsection
\ref{subsect: MP2} we define the extreme event of a sample, and
finally in
\ref{subsect: MP3} we present the test and prove the UMP property.\\
Section \ref{sect: application} shows the application of the test to
listeriosis. In particular, in Subsection \ref{subsect: listeriosi}
we present the data.
Then we perform a test to identify the place and time of
occurrence of the possible epidemic. This hypothesis test is
described in Subsection \ref{subsect: distr spatiotemp} and is based
on the UMP test introduced before; in this part of the article we
also recall a result to find $p$-values of a discrete distribution based
on the Skorokhod's representation of a random variable
\cite{williams}.

\section{A UMP test for extreme monotone randomized models}\label{sect:
MP_test}
To state the main result, we recall the definition and
properties of the Skorohod representation of a random variable \cite{williams}.\\
\begin{theorem} Let $F$ be a cumulative distribution function; then
$$N(\omega)=\sup\{y|F(y)<\omega\}$$
with $\omega$ from a probability space with uniform probability
distribution on $[0,1]$ is a random variable with cdf $F$.
\end{theorem}
\begin{definition}
The random variable defined in the previous theorem is called the
\textbf{Skorhod representation} of any random variable with
distribution $F$.
\end{definition}
The Skorohod representation of a random variable has many
properties; we mention two of them:
\begin{enumerate}
\item \label{uno} $F(N(\omega))\geq \omega$;
\item \label{due} $ F(z)>\omega \Rightarrow z>N(\omega)$.
\end{enumerate}

We are interested in finding a UMP test for the extreme event on a set of data. In our contest, the highest result we get, the more extreme it is. Therefore, to compare results from different distributions, we use the (randomized) quantile function as an index of ``extremeness''.

\subsection{Randomization of quantile function} \label{subsect: MP1}
Assume that we observe the real number $x$, that is the outcome of a continuous variable $X$ with cumulative function $F$. Its natural extremal index is
\(
p_x = P(X\leq x) =F(x).
\)
Unfortunately, it is well known that $F$ is continuous if and only if $F(X)$ is uniformly distributed on $(0,1)$, and in this case $F(X)$ is the quantile of $X$. We define now a randomization version of the quantile function which is uniformly distributed on $(0,1)$, even if the random variable $X$ is not continuous.

Let $F$ be a cumulative function. We define the function ${\Ff}_F:{\R}\times {[0,1]}\to [0,1]$ as
\begin{equation}\label{eq:def_Ff}
\Ff_F(x,u) = (1-u)F(x^-) + u F(x),
\end{equation}
so that $\Ff$ has the following properties:
\begin{itemize}
\item for any $(x,u)$, $F(x^-)\leq \Ff_F(x,u) \leq F(x)$, and hence $\Ff_F(x,u)$ is an extension of the function $F(x)$ when $F$ is not continuous;
\item if $F(x)>F(y)$, then $\Ff_F(x,u)>\Ff_F(y,v)$ for any $u,v\in (0,1)$; and hence $\Ff_F$ preserves the results with higher extremeness;
\item as a consequence of Properties~\ref{uno} and \ref{due} of the Skorhod representation,
if $U$ is a $(0,1)$-uniform random variable independent of $X$ (randomization effect), $\Ff_{F}(X,U)$ is always a $(0,1)$-uniform random variable, as the following lemma states.
\end{itemize}

\begin{lemma}\label{lem:unif}
Let $N$ be a random variable with distribution function $F$ and $U$ be a $(0,1)$-uniform random variable independent of $N$. Then the random variable
\[
\Ff_{F}({N},U) = (1-U)F({N}^-) + U F(N) ,
\]
is a $(0,1)$-uniform random variable.
\end{lemma}

\subsection{Extreme randomized event} \label{subsect: MP2}
A sample size of $n$ independent random variables $N_1,\ldots,N_n$
is given. Under the null hypothesis we assume
$\{F^0_i,i=1,\ldots,n\}$ to be their cumulative functions. Given a
set $U_1,\ldots,U_n$ of independent $(0,1)$-uniform random variables
(randomization effects), we may compute the indexes of extremeness
\begin{equation}\label{eq:def_Yi}
Y_i(N_i,U_i) = \Ff_{F^0_{i}}(N_i,U_i) =
(1-U_i)F^0_{i}(N_i^-) + U_i F^0_{i}(N_i), \qquad i = 1,\ldots, n.
\end{equation}
We observe the \emph{extreme event} $\hat{Y}=\max (Y_1,\ldots,Y_n)$.
By definition the extreme event is the greatest realization, once the random variables have been randomized and rescaled.
\begin{remark}
Note that, where $F^0_{i}$ is continuous at $N_i$, then $Y_i(N_i,U_i)= F^0_{i}(N_i)$: the randomization in \eqref{eq:def_Yi}
affects only the discrete set of outcomes of $N_i$ with positive probability.
\end{remark}

\subsection{Monotone models} \label{subsect: MP3}
We are interested in testing if the extreme event is coming from its
alternative distribution. More precisely, we test
\[
H_0: \{F_i = F_i^0,i=1,\ldots,n\}, \qquad H_1: \{F_i = F_i^0,i\neq j\}, F_j = F_j^1.
\]
The point of the test is the distribution of the maximum of the variables $\{Y_i,i=1,\ldots,n\}$. Under the null hypothesis, Lemma~\ref{lem:unif} states that $\{Y_i,i=1,\ldots,n\}$ are independent $(0,1)$-uniform random variables. Hence, the density of each $Y_i$ is constant if the distribution function of $N_i$ is $F_i^0$. The following definition states that in monotone models, the highest results are more and more likely in alternative hypothesis compared to the null one. In other words, $Y_i$ under $H_0$ is smaller than $Y_i$ under $H_1$ in the likelihood ratio order.

\begin{definition}\label{def:MLR}
We define the model to be \emph{monotone} if, for any $i=1,\ldots, n$, $F_{Y_i}$ is convex under the alternative hypothesis.
\end{definition}

\begin{remark}\label{rem:MLR}
Since every convex function on $[0,1]$ is differentiable almost everywhere with non-decreasing derivative, then a model is monotone if and only if $Y_i$ has a monotone non-decreasing density under the alternative hypothesis.
\end{remark}

\begin{theorem} \label{teo: MLR}
All the families that have the \emph{monotone likelihood ratio
(MLR)} property belong to monotone models.
\end{theorem}

\begin{proof}
Fixed $i\in \{1,\ldots,n \}$, let $F^0= F^0_i$, $F^1= F^1_i$ and $Y=Y_i$ as in \eqref{eq:def_Yi}. We denote with $N=N^1$ the fact that the true model is the alternative one $H_1$. We recall that the MLR property imply the absolute continuity of $F^1$ with respect to $F^0$ and viceversa.
We divide the proof between continuous and discrete models, since the contribution of $U$ in \eqref{eq:def_Yi} depends on it.
\begin{description}
\item[Absolutely continuous case:] in this case, by definition \(Y = F^0(N^1) \), where $N^1$ has density $f^1$ and $F^0$ is the cumulative function with density $f_0$. By the Change of Variables Formula, we get \(f_Y(y) = \frac{f^1(x)}{f^{0}(x)}\), where $y=F^0(x)$. The thesis is a consequence of the MLR property in continuous case, namely
\(
\frac{f^1(x_1)}{f^0(x_1)} \geq \frac{f^1(x_0)}{f^0(x_0)}
\),
for any \(x_1 > x_0\).
\item[Discrete case:] let $p\in(0,1)$ be fixed. Then there exists $x$ in the range of $N$ such that $p\in[F^0(x^-),F^0(x))$. By partitioning the space in $N^1< x$, $N^1> x$ and $N^1= x$, we obtain:
\begin{align*}
P(\Ff_{F^0}(N^1,U) \leq y )
=F^{1}(x^-) + \frac{p^1(x)}{p^0(x)} (y - F^0(x^-)) ,
\end{align*}
and hence \(f_Y(y) = \frac{p^1(x)}{p^{0}(x)}\).\\ Since $x$ is monotone in $p$, the thesis is a consequence of the MLR property in the continuous case, namely
\(
\frac{p^1(x_1)}{p^0(x_1)} \geq \frac{p^1(x_0)}{p^0(x_0)}
\),
for any \(x_1 > x_0\). \qedhere
\end{description}
\end{proof}

\begin{theorem} \label{teo: MP_test}
With the notations given above, a $\alpha$-level UMP test for testing
the extreme event of a monotone model is of the form
\[
\Phi(N_1,\ldots,N_n)
=
\begin{cases}
1, & \text{if } M > \sqrt[n]{1-\alpha};
\\
0, & \text{if } M  \leq \sqrt[n]{1-\alpha};
\\
1-\prod_{j\in R} \frac{\sqrt[n]{1-\alpha}-F^0_j(N_j^-)}{F^0_j(N_j)-F^0_j(N_j^-)}, & \text{otherwise};
\end{cases}
\]
where $M = \max (F^0_1(N_1^-),\ldots,F^0_n(N_n^-)) $ and $R=\{j \colon F^0_j(N_j^-)<\sqrt[n]{1-\alpha} <F^0_j(N_j)\}$, or, equivalently,
\[
\Phi(N_1,\ldots,N_n,U_1,\ldots,U_n)
=
\begin{cases}
1, & \text{if } \max (Y_1,\ldots,Y_n) > \sqrt[n]{1-\alpha};
\\
0, & \text{otherwise}.
\end{cases}
\]
\end{theorem}

\begin{proof}
The equivalence of the two definitions of $\Phi$ is a simple consequence of \eqref{eq:def_Yi}.

To use Neyman-Pearson lemma applied to the extreme event $\hat{Y}=\max (Y_1,\ldots,Y_n)$, we first note that, under $H_0$, $\{Y_i,i=1,\ldots,n\}$ are independent $(0,1)$-uniform random variables, and hence
\(
f^0_{\hat{Y}}(x) = n x^{n-1}
\)
for any $x\in (0,1)$, and moreover,
\[
E^0(\Phi(N_1,\ldots,N_n,U_1,\ldots,U_n))
=1 - (\sqrt[n]{1-\alpha})^n = \alpha.
\]
Under the null hypothesis $H_0$, setting $\tau_j = P(\hat{Y}=Y_j)$, we get
\[
P_{H_1}(\hat{Y}\leq x)
= \sum_j P_{H_1}(\hat{Y}\leq x| \hat{Y}=Y_j) P(\hat{Y}=Y_j)
= \sum_j  x^{n-1}F^1_{Y_j}(x)  \tau_j ;
\]
and hence
\begin{multline*}
\frac{f^1_{\hat{Y}}(x) }{f^0_{\hat{Y}}(x)} = \frac{\sum_j  (x^{n-1}f^1_{Y_j}(x)+(n-1)x^{n-2}F^1_{Y_j}(x))  \tau_j }{nx^{n-1}} =\\
= \sum_j  \frac{f^1_{Y_j}(x) }{n} \tau_j
+ \frac{n-1 }{n}\sum_j  \frac{\int_0^x f^1_{Y_j}(y)\,dy  }{x} \tau_j ,
\end{multline*}
and, by definition of monotone model, both the terms are convex
combination of monotone non-decreasing functions, the second being
the integral mean of a monotone and non-negative function.
Therefore, $\frac{f^1_{\hat{Y}}(x) }{f^0_{\hat{Y}}(x)} $ is monotone
in $x$, and the thesis is proved.
\end{proof}

\section{Application to listeriosis} \label{sect: application}
 Invasive
listeriosis is a rare severe disease with low annual incidence ($<
1/100\,000$). It typically includes long incubation periods (7-60
days), usually resulting in hospitalization (85\% to 90\%) and has a
high fatality rate (20-50\%). Persons with specific
immunocompromising conditions, pregnant women and newborns appear to
be particularly susceptible to invasive listeriosis, and most
reported cases occur in these specific risk groups. The
identification of outbreaks is difficult because of the long
incubation period of the invasive forms (even several weeks) and of
the probable large number of asymptomatic or paucisymptomatic
infections even in people exposed to the same infection vehicle
\cite{list_primer, list_review}.

\subsection{The data: listeriosis in Lombardy} \label{subsect: listeriosi}
The data we have collected and analyzed consist of detailed
information about the persons who have contracted listeriosis in
Lombardy between years 2005 and 2011. This region accounts for 16\%
of the Italian population ($\sim 10\,000\,000$ inhabitants), but for
55\% of the notified listeriosis cases in the entire country. These
cases have been identified through a laboratory-based surveillance
system enhanced in the latest years \cite{mammina}. We have focused
our attention on some variables, such as the date of identification
of the disease and the province of residence of the patient, so that
we are able to analyze the spatiotemporal location of cases.
We notice that the data increase in the latest years;
this fact is due to an improvement in the
transmission of information: since 2008 the process has become more
systematic. Hence we have decided to limit our statistical tests to
the cases individuated from 2008 on.


Another important variable of the data is the molecular type of each
\emph{L.mo\-nocyto\-ge\-nes} isolate, which has been identified through a
laboratory analysis based on MLST (MultiLocus Sequence Typing)
\cite{salcedo}. Thanks to this laboratory work, it has been possible
to concentrate our statistical study on a single sequence type (ST).
In fact possible confirmations of the presence of epidemics would
make sense only if the cases refer to a unique type \cite{sauders}.\\
The statistic tests we have performed consider only  the data
referred to isolates belonging to ST38, which is the most numerous
one. In fact the database contains information about 180 cases, of
which 139 are notified since 2008; since this year there are 36
strains belonging to ST38, whereas the second most numerous is ST1,
with only 18 cases.

\subsection{Identification of epidemics in space and time} \label{subsect: distr
spatiotemp}

We ask ourselves if there is evidence of epidemics in our data.
We test the null hypothesis of absence of epidemics ($H_0$)
against the alternative hypothesis of presence of epidemics ($H_1$).
In particular, if in a certain spatiotemporal region an epidemic occurs, the number of
detected cases increases.\\
Let $N_{i,j}$ be the number of detected cases in region $i$ at time $j$;
we cannot use the statistics $\max N_{i,j}$ because these random variables
are not identically distributed. In fact under the null hypothesis we suppose
that the number of detected cases is distributed as a Poisson variable
with intensity $\lambda\cdot p_{i,j}$
depending on the population of regions $i$ at times $j$,
and hence we test this hypothesis
with the UMP test given in Theorem~\ref{teo: MP_test}.
We use a conservative
estimate of $\lambda$ ($\hat{\lambda}\approx 9.703 \cdot 10^{-7}$).

Each case belonging to ST38 has
been provided with a spatial variable defining the province of
residence of the patient. We point out that the provinces of Sondrio and Mantova have not
communicated any case of listeriosis and so they have been excluded
from the analysis: $R$ is then a set describing Lombardy without the
territories of these two provinces. Besides, we have to specify that province
Monza e Brianza was born during the considered period of time, so we
have decided to attribute label ``MB'' to any patient living in
places belonging to this province in 2011, even if the case of
listeriosis was detected before the birth of the province. The time
interval $T$ has been partitioned through 4 years:
$2008,\ldots,2011$. Table
\ref{tab:spaziotemp} shows the values of $n_{i,j}$ for any $1 \leq i\leq 10, 1 \leq
j\leq 4$.

\begin{table}
\caption{Number of cases $n_{i,j}$ in each province and
year}
{\begin{tabular}{@{}rlcccccccccc}
\hline
& & \multicolumn{10}{c}{\bf{Province}} \\
&  & BG & BS & CO & CR & LC & LO & MB & MI & PV & VA\\
\hline
{\multirow{4}{*}{\begin{sideways}\bf{Year}\end{sideways}}}
& 2008 & 0 & 0 & 0 & 0 & 0 & 0 & 0 & 0 & 0 & 1\\
& 2009 & 2 & 1 & 0 & 1 & 1 & 1 & 1 & 3 & 0 & 0\\
& 2010 & 8 & 1 & 1 & 0 & 0 & 0 & 0 & 4 & 0 & 1\\
& 2011 & 4 & 0 & 0 & 0 & 0 & 0 & 0 & 4 & 1 & 0\\
\hline
\end{tabular}}
\label{tab:spaziotemp}
\end{table}

The values of $p_{i,j}$ are
given by ISTAT \cite{istat}; for each year we have chosen the data
referring to December 31st. As concerns years 2008 and 2009, we have
chosen as population of Monza e Brianza the same population of
January 1st 2010, and this value has been subtracted to the
population of Milano.\\

We define
$$
Y_{i,j}:= \Ff_{{F^0_{i,j}}}(N_{i,j},U_{i,j}) = (1-U_{i,j})F_{F^0_{i,j}}(N_{i,j}^-) + U_{i,j} F_{F^0_{i,j}}(N_{i,j})
$$ for any $1
\leq i\leq 10, 1 \leq j\leq 4$. By Theorem~\ref{teo: MP_test}, we focus on the statistics
\(
M:=
{\max_{i, j}} Y_{i,j}.
\)
To find the $p$-value of our test,
our aim is to calculate an upper and lower bound in terms of the observed
$F^0_{i,j}(N_{i,j})$ for
\[
P({\max_{i, j}} Y_{i,j}>\max_{i,j} \Ff_{{F^0_{i,j}}}(n_{i,j},u_{i,j}) ) = P(M>m ) .
\]
If this probability is lower than the
confidence level of our test, then we can reject the null
hypothesis. \\

By definition of $\Ff_{F}$, and since $N_{i,j}$ is integer-valued ($N_{i,j}^-=N_{i,j}-1$), we trivially have
\[
F_{F^0_{i,j}}(N_{i,j}-1)  \leq Y_{i,j} \leq F_{F^0_{i,j}}(N_{i,j}) , \forall i,j
\]
and hence
\begin{equation}\label{eq:trivial_MAX}
\underline{M}:=\max_{i,j} F_{F^0_{i,j}}(N_{i,j}-1)  \leq M \leq \max_{i,j} F_{F^0_{i,j}}(N_{i,j}) =: \overline{M}.
\end{equation}
If we define
\begin{equation*}
\underline{m}:=\max_{i,j} F_{F^0_{i,j}}(n_{i,j}-1)  , \qquad \overline{m} = \max_{i,j} F_{F^0_{i,j}}(n_{i,j}) .
\end{equation*}
then, by \eqref{eq:trivial_MAX}, under $H_0$ we get
\[
1 - \overline{m}^{i\cdot j} = P(M>\overline{m}) \leq P(M>m) \leq P(M>\underline{m}) = 1 - \underline{m}^{i\cdot j},
\]
i.e., $1 - \overline{m}^{i\cdot j} \leq p \leq 1 - \underline{m}^{i\cdot j}$, where $p$ is the $p$-value of our UMP test.

If $1 - \underline{m}^{i\cdot j}$
is smaller than our significance level, we can reject the null
hypothesis and state that an epidemic occurred ($\Phi=1$ in
Theorem~\ref{teo: MP_test}); if $1 - \overline{m}^{i\cdot j}$
is greater than the
significance level, the null hypothesis cannot be rejected ($\Phi=0$
in Theorem~\ref{teo: MP_test}); if only the first value is smaller
than the significance level, a randomized test has to be carried on
($0<\Phi<1$ in Theorem~\ref{teo: MP_test}).
With our sample we find that
\( 1 - \overline{m}^{i\cdot j} = 0.00006 \) and \( 1 - \underline{m}^{i\cdot j} = 0.00053\).
These values force us to reject the null hypothesis, and hence the
number of cases of listeriosis with isolates belonging to ST38
detected in the province and year corresponding to the maximum is significantly higher
than expected under non-epidemic conditions. Hence we can statistically conclude
that an epidemic has occurred in Bergamo in 2010.\\
We have also continued the analysis by asking ourselves whether in
some other provinces and years we could find some epidemics. To this
aim we have repeated the spatiotemporal test excluding from the
sample the datum that refers to Bergamo cases in 2010. This analysis
has not given any result, because in no case we have obtained
sufficiently small values to reject the null hypothesis. This
conclusion does not mean that we exclude the possibility of
existence of other epidemics, but just that further analyses have to
be carried on.

\end{document}